\documentclass[a4paper,UKenglish,cleveref, autoref, thm-restate]{lipics-v2021}
\usepackage[linesnumbered,ruled,vlined]{algorithm2e}

\newcommand{\samp}{\textsf{SAMP}~}
\newcommand{\cond}{\textsf{COND}~}
\newcommand{\eval}{\textsf{EVAL}~}
\newcommand{\pcond}{\textsf{PCOND}~}
\newcommand{\intcond}{\textsf{INTCOND}~}

\newcommand{\accept}{\textsf{Accept}~}
\newcommand{\reject}{\textsf{Reject}~}
\newcommand{\compare}{\emph{Compare}~}

\title{Distribution Testing Meets Sum Estimation} 
\author{Pinki Pradhan}{National Institute of Science Education and Research, An OCC of Homi Bhabha National Institute, Bhubaneswar, Odisha, India}{pinki.pradhan@niser.ac.in}{}{}
\author{Sampriti Roy}{Department of Computer Science and Engineering, IIT Madras, Chennai}{cs18d200@smail.iitm.ac.in}{}{}
\authorrunning{P. Pradhan, S. Roy}

\Copyright{Jane Open Access and Joan R. Public} 



\category{} 

\relatedversion{} 

\EventEditors{John Q. Open and Joan R. Access}
\EventNoEds{2}
\EventLongTitle{42nd Conference on Very Important Topics (CVIT 2016)}
\EventShortTitle{CVIT 2016}
\EventAcronym{CVIT}
\EventYear{2016}
\EventDate{December 24--27, 2016}
\EventLocation{Little Whinging, United Kingdom}
\EventLogo{}
\SeriesVolume{42}
\ArticleNo{23}

\begin{document}

\maketitle
\nolinenumbers
\begin{abstract}

 We study the problem of estimating the sum of \( n \) elements, each associated with a weight $ w(i) $, within a structured universe. Specifically, our goal is to estimate the sum $W = \sum_{i=1}^n w(i)$ within a factor of $(1 \pm \epsilon)$ using a sublinear number of samples, under the assumption that the weights are non-increasing, i.e., $ w(1) \geq w(2) \geq \dots \geq w(n) $. The sum estimation problem is well-studied under different access models to the universe $U$. 
 However, to the best of our knowledge, nothing is known about the sum estimation problem using non-adaptive conditional sampling. To address this gap, we explore the sum estimation problem using non-adaptive conditional weighted and non-adaptive conditional uniform samples, assuming that the underlying distribution ($D(i)=w(i)/W$) is monotone. 
 Furthermore, we extend our approach to the case where the underlying distribution of $U$ is unimodal. Additionally, we investigate the problem of approximating the support size of $U$ when $w(i) = 0$ or $w(i) \geq W/n$ using hybrid sampling, where both weighted and uniform sampling are available for accessing $U$.

We provide an algorithm for estimating the sum of $n$ elements where the weights of the elements are non-increasing. Our algorithm requires $O(\frac{1}{\epsilon^3}\log{n}+\frac{1}{\epsilon^6})$ non-adaptive weighted conditional samples and $O(\frac{1}{\epsilon^3}\log{n})$ non-adaptive uniform conditional samples.
Our algorithm also follows the $\Omega(\log{n})$ lower bound proposed by \cite{ACK15}. We also extend our algorithm when the underlying distribution $D$ is unimodal. The sample complexity of the algorithm is the same as that of monotone with an additional $O(\log{n})$ adaptive evaluation queries to find the minimum weighted point in the domain $[n]$. Additionally, we investigate the problem of estimating the support size of $U$, which consists of $n$ elements such that the weight corresponding to them is either $0$ or at least $W/n$. The algorithm uses $O\big( \frac{\log^3{(n/\epsilon)}}{\epsilon^8}\cdot \log^4{\frac{\log{(n/\epsilon)}}{\epsilon}}\big)$ uniform samples and $O\big( \frac{\log{(n/\epsilon)}}{\epsilon^2}\cdot \log{\frac{\log{(n/\epsilon)}}{\epsilon}}\big)$ weighted samples from the universe and approximate the support size $k$ such that $k-2\epsilon n\leq\hat{k}\leq k+\epsilon n$.

\end{abstract}

\section{Introduction}
In recent years, distribution testing and sum estimation have attracted considerable attention due to their applications in statistics and data science. Distribution testing is a subfield of property testing that aims to determine whether an unknown distribution satisfies a certain property \(\mathcal{P}\) or if it is \(\epsilon\)-far from \(\mathcal{P}\). In this context, the distribution is accessed through samples, and the key measure of interest is the number of samples needed to test the property effectively (see \cite{Fischer01}, \cite{Goldreich10}, \cite{Goldreich17}, \cite{BY22} for comprehensive overviews).

On the other hand, sum estimation involves efficiently approximating the total sum of a set, where each element is assigned a weight. Let $U$ be a universe of $n$ elements, with each element $i \in [n]$ associated with a weight $w(i): U \to [0, \infty)$. We do not have access to the entire universe $U$, and the goal is to compute the sum $W = \sum_{i=1}^{n} w(i)$ by sampling elements from $U$. Each time we get a sample, we observe one element along with its weight to estimate $W$ using a sublinear number of samples. Since exact computation of $W$ requires $\Omega(n)$ many samples, we aim to provide an estimated value $\hat{W}$ such that for any $\epsilon \in (0, 1)$, $(1-\epsilon)W \leq \hat{W} \leq (1+\epsilon) W$ with probability at least $\frac{2}{3}$. The universe of elements $U$ can be viewed as an unknown distribution $D$. Sampling an element from $U$ is equivalent to drawing an index $i$ independently and identically distributed according to $D$ in the standard access model. For example, the underlying distribution of the universe can be thought of as $D(i) = \frac{w(i)}{W}$ and sampling an element from $U$ according to its weight is equivalent to sampling an element $i$ according to the standard access oracle of $D$. Each sample we get is a pair containing the index and the weight of the element corresponding to this index, i.e., $(i, w(i))$.  

The sum estimation problem was first introduced by Motwani et al. \cite{MPX07}, where the authors developed a weighted sampling method to estimate \(W = \sum_{i \in [n]} w(i)\). This was later improved by Beretta and  T\v{e}tek \cite{BT24}. However, an open question remains whether it is possible to further optimize the sample complexity under certain assumptions about the universe. In addition to fundamental problems such as testing uniformity, identity, and closeness, distribution testing has been extended to cases where the distribution is assumed to have a specific structure. As we consider an underlying distribution $D$ corresponding to the universe's weights, it is worth checking whether the sample complexity of the classical sum estimation problem can be improved under some assumptions about the structure of $D$.

Another interesting connection between distribution testing and sum estimation is exploring whether approximation problems in distribution testing can be applied to the universe of a set of elements when the weights are unknown. For example, consider the problem of estimating the support size of a distribution. Let \(D\) be a distribution over \([n]\), and the support size of \(D\) is defined as the number of non-zero elements in the domain \([n]\). A similar question arises when considering a universe of \(n\) elements, where we aim to estimate the support size without knowing the weights.

\subsection{Sampling Oracles}
The primary goal of a distribution testing algorithm is to test the properties of an unknown distribution without accessing the entire distribution. Similarly, in the case of sum estimation, the goal is to estimate the total weight of a set of elements by looking at the entire set. To achieve this, one requires oracle access to the unknown object, allowing the algorithm to sample a subset and use it as input for testing and estimation. Whether drawing samples from a particular probability distribution or from a universe of $n$ elements, both approaches utilize similar strategies. One of the most intuitive models for distribution testing involves sampling independently and identically distributed indices from an unknown distribution. More specifically, for an unknown distribution $D$ over $[n]$, the model is defined as \emph{standard access model} (\samp), which returns an index $i\in[n]$ with probability $D(i)$. Similarly, the most fundamental way of sampling an index from a universe $U$ over $[n]$ is to sample $i\in[n]$ according to the weight $w(i)$. This technique was introduced by the seminal work of Motwani et al. \cite{MPX07}. Intuitively, in both models, the more the weight of a point, the more likely we are to sample that point. Therefore, if $D$ is an underlying distribution corresponding to the set of elements in $U$, sampling an index according to the \samp is equivalent to sample an element with probability $w(i)/W$, where $W=\sum_{i\in[n]}w(i)$.

Additionally, both fields support adaptive and non-adaptive conditional sampling. To this end, Canonne et al. \cite{CFGM13} and Chakraborty et al. \cite{CRS15} introduced the conditional sampling model, which has proven to be a powerful tool for testing the properties of probability distributions more efficiently. For a fixed distribution $D$ over the domain $[n]$, a conditional oracle for $D$, \cond is defined as follows: the oracle takes as input a query set $S \subseteq [n]$ chosen by the algorithm, that has $D(S)> 0$ (when $D(S)=0$, an i.i.d sample is generated). The oracle returns an element $i \in S$, where the probability that element $i$ is returned is $D_S(i) = D(i)/D(S)$. The model is defined \intcond when the conditioning is done over an interval. Similarly, in the case of sum estimation, Acharya et al. \cite{ACK15} introduced the conditional sampling technique. For a specific set $S\subset [n]$, the oracle returns an element $i\in S$ with probability $w(i)/W(S)$, where $W(S)=\sum_{i\in S}w(i)$. However, given the set $S$, in both cases, we are also allowed to sample an element uniformly from $S$. 

In distribution testing, a general model of query that comes naturally is to sample a set of points from the domain and explicitly get their probabilities. This motivates a new model called evaluation oracle (\eval) proposed by Rubinfeld and Servedio \cite{RS05}. In particular, an evaluation oracle (\eval) for $D$ is defined as follows: the oracle takes as input a query element $i\in[n]$ and returns the probability weight $D(i)$ that the distribution puts on $i$. Similarly, we extend this model when there are $n$ weighted elements instead of a distribution, and upon querying an input $i\in[n]$, the oracle returns $w(i)$. The \eval model can be further enhanced by combining it with additional access to the standard sampling model. Specifically, this gives rise to the \emph{dual model}, where points are sampled according to the standard model, and then the sampled points are queried using the \eval oracle. This approach motivates sampling a point \( i \) from the universe based on its weight \( w(i) \) and explicitly querying for \( w(i) \). In other words, we can treat this as a dual model where, upon querying, we receive a pair \( (i, w(i)) \) with probability \( w(i)/W \). In the conditional framework, this is defined so that the pair \( (i, w(i)) \) is returned with probability \( w(i)/W(S) \). This model was explored in the work of Acharya et al. \cite{ACK15}. Similarly, in the uniform conditional setting, the pair is returned with probability \( 1/|S| \).

In this work, we investigate the sum estimation problem under a specific condition on the universe \( U \). Specifically, we focus on cases where the weights of the elements in \( U \) are non-increasing and propose an algorithm to estimate the sum of the elements in \( U \). To tackle this, we adopt the conditional access model. In this model, we use two sampling strategies: (1) given a set \( S \), we sample a pair \( (i, w(i)) \) uniformly at random from \( S \), meaning \( (i, w(i)) \) is returned with probability \( 1/|S| \); and (2) we sample a pair \( (i, w(i)) \) with probability \( w(i)/W(S) \), which is identical to \emph{dual model} for $S\subseteq [n]$. Using these techniques, we also extend our approach for estimating the sum of the elements in \( U \), particularly when the underlying distribution is unimodal. A unimodal distribution is a probability distribution where the weights of the elements decrease (or increase) up to a point $j\in[n]$ and then start increasing (or decreasing) up to $n$. To find the valley or the minimum weight element, we use the binary search approach, where check the weight of the middle element, compare it with the previous and next element, and then recurse on the one half based on the weights. In this process, we do not sample anything from $U$, rather we ask the the weight $w(i)$ for an input $i$. These queries are identical to making \eval queries, which is less powerful than \emph{dual model}. Hence, for the sake of clarity, we define the \emph{evaluation model} for $U$, where given an input $i\in[n]$, the oracle returns $w(i)$.

Additionally, we explore the problem of estimating the support size of the $U$, where support $k\subset [n]$ is defined by the number of non-zero elements in $U$. In this case, we use the hybrid model of computation, which allows us to sample a pair $(i,w(i))$ $(1)$ uniformly from $[n]$, and $(2)$ according to the probability $w(i)/W$.

\subsection{Background and Related Works}
Distribution testing was introduced by the seminal work of Goldreich and Ron (\cite{GR00}), who gave a sample-efficient algorithm for testing if an unknown distribution is uniform or $\epsilon$-far from uniform. Later, a study by Batu et al. (\cite{BFF+01}) addressed the problem of testing whether two unknown random variables share the same distribution (known as closeness or equivalence testing). This initiated a line of research focusing on testing uniformity, identity, closeness (\cite{Pan08}, \cite{BFR+10}, \cite{Valiant}, \cite{VV14}, \cite{CDVV14}, \cite{DKN15b}, \cite{DKN15a})), structure of distributions (monotonicity, $k$-modality, decomposibility etc) (\cite{BKR04}, \cite{DDS12}, \cite{CDSS13}, \cite{DDSVV13}, \cite{CDGR16}, \cite{FLV19}) and many more using the standard access model when i.i.d samples are drawn from an unknown distribution. Another interesting problem in this domain is to determine whether testing becomes easier when there is a predefined structural property. Daskalakis et al. \cite{DDSVV13} studied the identity, and closeness testing problems considering the unknown distributions to be monotone. Another work by Diakonikolas et al. \cite{DKN15} focused on testing identity for a set of distributions which are promised to be $t$ piecewise constant, log-concave, monotone hazard rate, etc. Rubinfeld and Vasiliyan \cite{RV20} studied the problem of learning a distribution over $\{0,1\}^n$ when it is known to be monotone.

Distribution testing has also been addressed under the conditional sampling model. Canonne et al. \cite{CRS15} studied the problem of testing uniformity, identity, and closeness using the power of the conditional sampling model. Later, Canonne \cite{Can15} proposed an algorithm for testing monotonicity using the same model. Canonne and Rubinfeld \cite{CR14} extensively studied distribution testing problems using the evaluation oracle and its variants. They focused on fundamental problems such as uniformity testing, identity testing, and closeness testing and proposed efficient algorithms for each of these tasks in the \eval model. Estimating the support size of a distribution has been a central problem in this field as well. Raskhodnikova et al. \cite{RRSS09} addressed the problem of approximating the support size of a distribution from a small number of samples when each element in the distribution appears with probability at least $1/n$. Later, Valiant and Valiant \cite{VV10} proved a lower bound under the standard access model for estimating the support size without any assumption on the weight. Acharya et al. \cite{ACK14} proved a non-adaptive lower bound for support size estimation using conditional queries. More recently, Chakraborty et al. \cite{CKM23} investigated the longstanding gap between the upper bound for support size estimation under conditional sampling model. Additionally, Canonne and Rubinfeld \cite{CR14} also studied this problem under \eval model with the assumption that the probability mass of each element of the domain is either zero or at least $1/n$. This line of research motivates us to explore the problem of estimating the support size when there is a universe of $n$ elements instead of a distribution.

Estimating the sum of $n$ variables is also a fundamental problem in statistics. The problem becomes simple when the weights of the elements are in the range $[0,1]$. Canetti et al. \cite{CEG95} proved a lower bound for this problem using uniform sampling. Motwani et al.\cite{MPX07} first designed a sublinear algorithm for this problem and gave an estimator for the sum using $\tilde{O}(\sqrt{n}/\epsilon^{7/2})$ weighted samples. They also studied the sum estimation problem in hybrid sampling, where in addition to weighted sampling, one can perform uniform sampling, and they gave an estimator for the sum using $\tilde{O}(\sqrt{n}/\epsilon^{9/2})$ hybrid samples. Later, Beretta and T{\v{e}}tek \cite{BT24} improved this sample complexity and gave an algorithm for estimating the sum using $\Theta(\sqrt{n}/\epsilon)$ weighted samples. They also provided an algorithm for sum estimation using $O(n^{1/3}/\epsilon^{4/3})$ hybrid samples.
Acharya et al. \cite{ACK15} explored the sum estimation using adaptive conditional sampling. They gave an algorithm to compute $\hat{W}$ by using  $O(\frac{\log^3 n \log \log n }{\epsilon^2})$ many adaptive conditional samples.  The authors used the conditional sampling over intervals which aligns well with \intcond model of distribution testing, building a strong bridge between distribution testing and sum estimation. Moreover, their work builds on the adaptive conditional queries. In contrast, we give an algorithm for sum estimation using non-adaptive conditional sampling, which is a weaker model as compared to adaptive conditional sampling with the assumption that the weights of the universe is monotonically non-increasing universe. The problem of approximating the number of distinct elements in a sequence of length $n$ has been studied by Charikar et al. \cite{CCMN00} and Bar-Yossef et al. \cite{BKS01}. These problems are closely related to the approximation of support size and motivate us to further study the same under different query models.

\subsection{Our Contributions}
We explore the problem of sum estimation when there is a specific structure of the set of elements. In particular, let $U$ defines a set of $[n]$ elements each associated with a weight $w(i)$, such that $w(1)\geq w(2)\geq...\geq w(n)$. In this setting, the goal is to find a $(1\pm \epsilon)$ approximation of the weight $W=\sum_{i\in[n]} w(i)$. We use the power of conditional weighted as well as conditional uniform sampling. Specifically, given a set $S\subseteq [n]$, the oracle returns a pair $(i,w(i))$ with probability $w(i)/W(S)$ in case of weighted conditional query. Additionally, we can also sample a pair $(i,w(i))$ uniformly at random from $S$. We go beyond the conditional sampling and address the problem of support size estimation considering hybrid model which returns a pair $(i,w(i))$ with probability $w(i)/W$ as well as a pair $(i,w(i))$ uniformly at random from $[n]$. Our key contributions are as follows:
\begin{itemize}
    \item We provide an algorithm for estimating the sum of $n$ elements where the weights of the elements are non-increasing. Our algorithm requires $O(\frac{1}{\epsilon^3}\log{n}+\frac{1}{\epsilon^6})$ weighted conditional samples and $O(\frac{1}{\epsilon^3}\log{n})$ uniform conditional samples. Our algorithm also follows the $\Omega(\log{n})$ lower bound proposed by \cite{ACK15}. We also extend our algorithm when the underlying distribution corresponding to the weights is unimodal. In particular, we estimate the sum of $[n]$ elements when the weights decrease up to a point $j\in[n]$ and then increase up to $n$. The sample complexity of the algorithm is the same as that of monotone with an additional $O(\log{n})$ adaptive \eval queries to find the minimum weighted point in the domain $[n]$. 
    \item We investigate the problem of estimating the support size of $U$, which consists of $n$ elements such that the weight corresponding to them is either $0$ or at least $W/n$. The algorithm uses $O\big( \frac{\log^3{(n/\epsilon)}}{\epsilon^8}\cdot \log^4{\frac{\log{(n/\epsilon)}}{\epsilon}}\big)$ uniform samples and $O\big( \frac{\log{(n/\epsilon)}}{\epsilon^2}\cdot \log{\frac{\log{(n/\epsilon)}}{\epsilon}}\big)$ weighted samples from the universe and approximate the support size $k$ such that $k-2\epsilon n\leq\hat{k}\leq k+\epsilon n$. 
\end{itemize}

\section{Notation and Preliminaries}
We use $U$ as a universe with set of $n$ elements each associated with weight $w(i)$, and $W=\sum_{i\in [n]}w(i)$. We also define $D$ to be the underlying probability distribution over $[n]$ such that for $i\in [n]$, $D(i)=w(i)/W$. For a set $S\subseteq [n]$ $D_S$ is defined as a conditional probability distribution over $S$, such that for $i\in [S]$, $D_S(i)=w(i)/W(S)$, where $W(S)$ denotes $\sum_{i\in [S]} w(i)$. $\mathcal{U}_S$ defines the uniform distribution over $[S]$. Let $D_1$ and $D_2$ be two distributions over $[n]$, the total variation distance is denoted by $d_{TV}(D_1,D_2)=\frac{1}{2} \sum_{i\in [n]} |D_1(x)-D_2(x)|$. Let $\mathcal{D}$ be the set of all probability distributions supported on $[n]$. A property $\mathcal{P}$ is a subset of $\mathcal{D}$. We say that a distribution $D$ is $\epsilon$ far from $\mathcal{P}$, if $D$ is $\epsilon$ far from all the distributions having the property $\mathcal{P}$. That is, $d_{TV}(D,D') > \epsilon$ for every $D' \in \mathcal{P}$.  A distribution $D$ over $[n]$ is \emph{monotone} (non-increasing) if its probability mass function satisfies $D(1) \geq D(2) \geq \dots \geq D(n)$. Let $D$ be a distribution over the domain $[n]$, and there exists a set of partitions of the domain into $\ell$ disjoint intervals, $\mathcal{I}=\{I_j\}_{j=1}^{\ell}$. The flattened distribution $(D^f)^{\mathcal{I}}$ corresponding to $D$ and $\mathcal{I}$ is a distribution over $[n]$ such that, for $j\in [\ell]$ and $i\in I_j$; $(D^f)^{\mathcal{I}}(i)=\frac{\sum_{t\in I_j}D(t)}{|I_j|}$.  Monotone distributions have a special property called \emph{oblivious partition} proposed by \cite{Birge} which depicts that any monotone distribution is close to its flattened distribution with respect to some oblivious partitions of intervals with exponential size. It can be defined formally as follows,
\begin{theorem}[Oblivious partitioning \cite{Birge}] \label{thm: berge_1}
Let $D$ be a non-increasing distribution over $[n]$ and $\mathcal{I}=\{I_1,...,I_{\ell}\}$ is an interval partitioning of $D$ such that $|I_j|=(1+\epsilon)^j$, for $0<\epsilon<1 $, then $D$ has the following properties,
\begin{itemize}
    \item $\ell=O(\frac{1}{\epsilon}\log{(n\epsilon)})$ 
    \item The flattened distribution corresponding to $\mathcal{I}$, $(D^f)^{\mathcal{I}}$ is close to $D$.
\end{itemize}
\end{theorem}
Similarly, when the distribution $D$ is non-decreasing, the oblivious partitioning can be defined as follows,
\begin{theorem}[Oblivious partitioning \cite{Birge}] \label{thm: berge_2}
Let $D$ be a non-decreasing distribution over $[n]$ and $\mathcal{I}=\{I_1,...,I_{\ell}\}$ is an interval partitioning of $D$ such that $|I_j|=(1+\epsilon)^j$, for $0<\epsilon<1 $, where $j$ is considered in decreasing order from $\ell,\ell-1...$ to $1$, then $D$ has the following properties,
\begin{itemize}
    \item $\ell=O(\frac{1}{\epsilon}\log{(n\epsilon)})$
    \item The flattened distribution corresponding to $\mathcal{I}$, $(D^f)^{\mathcal{I}}$ is close to $D$, i.e, $d_{TV}(D,(D^f)^{\mathcal{I}})\leq \epsilon$.
\end{itemize}
\end{theorem}
The following inequality will be used in the analysis of the Algorithm \ref{alg:summono}.
\begin{lemma}[Hoeffding's inequality] \label{lem: hoeffding}
Let $X_1,...,X_m\in[a_i,b_i]$ be independent random variables, and $X=(X_1+X_2+...+X_m)/m$. Then $Pr[|X-\mathbb{E}[X]|\geq t]\leq 2exp(\frac{-2(mt)^2}{\sum_{i}(b_i-a_i)^2})$.
\end{lemma}
\section{Estimating Sum for Monotonically Non-increasing universe using Conditional Samples}
Given an universe of size $[n]$, where each element is associated with a weight $w(i)$ and the weights are monotonically non-increasing, the goal is to estimate the weight $W=\sum_{i=1}^nw(i)$ such that the estimated weight $(1-\epsilon)W\leq \hat{W}\leq (1+\epsilon)W$. We use the concept of conditional sampling in our approach. In particular, given a subset $S\subseteq [n]$, when queried a pair $(i,w(i))$ will be returned with probability $\frac{w(i)}{W(S)}$, where $W(S)$ defines the total weight of the set $S$. Additionally, we use uniform conditional samples which returns a pair $(i,w(i))$ with probability $\frac{1}{|S|}$. For a subset $S\subseteq [n]$, we can consider a conditional distribution $D_S$ over $[S]$, such that $D_S(i)=\frac{w(i)}{W(S)}$. We start with an algorithm for testing whether $D_S$ (over $[S]$) is uniform or $\epsilon$-far from uniformity using conditional sample on $S$.

\begin{algorithm}[hbt!]
\caption{Testing Uniformity}
\label{alg:uniformity}
\KwIn{Sample access to $D_S$, where $S\subseteq [n]$, error parameter $\epsilon\in(0,1)$,}
\KwOut{Decision (\accept/\reject)}
Sample $T_1 = O\left(\frac{1}{\epsilon}\log{(1/\epsilon)}\right)$ pairs $(i, w(i))$ from $D_S$\;
Sample $T_2 = O\left(\frac{1}{\epsilon}\log{(1/\epsilon)}\right)$ pairs $(j, w(j))$ uniformly at random from $S$\;
\For{each pair $(i, j)$, where $i \in T_1$ and $j \in T_2$}{
    \If{$\frac{w(i)}{w(j)} > \left(1 + \frac{\epsilon}{2}\right)$}{
        \reject;
        \textbf{EXIT}\;
    }
}
\accept;

\end{algorithm}

\begin{theorem}
    The algorithm Testing Uniformity samples $O(\frac{1}{\epsilon}\log{(1/\epsilon)})$ points from $D_S$ and $O(\frac{1}{\epsilon}\log{(1/\epsilon)})$ points uniformly from $S$ and does the following,
    \begin{itemize}
        \item Returns \accept with probability at least $2/3$ if $D_S$ is uniform over $S$
        \item Returns \reject with probability at least $2/3$ if $d_{TV}(D_S,\mathcal{U}_S)>\epsilon$
    \end{itemize}
\end{theorem}
\begin{proof}
    The proof is similar to the analysis of the uniformity testing algorithm presented by \cite{CRS15}. The authors used the \pcond queries to distinguish a uniform distribution from far from uniformity. Unlike them, we use the weight of each point $w(i)$ for testing uniformity as our query model returns a pair $(i,w(i))$. Let $D_S$ is uniform over $[S]$, then for every pair $(i,w(i))$ and $(j,w(j))$, $\frac{w(i)}{w(j)}=1$. Hence the algorithm will not output \reject inside the for loop.

Let $d_{TV}(D_S,\mathcal{U}_S)>\epsilon$. We start by defining the following sets,

$H=\Big \lbrace h \in [S] | D_S(h)\geq \frac{1}{S} \Big\rbrace$ and
$L=\Big\lbrace l \in [S] | D_S(l) < \frac{1}{S} \Big\rbrace$.
It is easy to see that when $D_S$ is $\epsilon$-far from uniformity, $\sum_{h\in H}(D(h)-\frac{1}{S})=\sum_{l\in L}(\frac{1}{S}-D(l))>\epsilon/2$. Now, we have the following observations:
\begin{itemize}
    \item Let $L'\subset L$ such that $L'=\Big\lbrace l' \in L | D_S(l') < \frac{1}{S}- \frac{\epsilon}{4S}\Big\rbrace$. Then $|L'|>(\epsilon/4)S$. While sampling $O(\frac{1}{\epsilon}\log{(1/\epsilon)})$ points uniformly from $S$, the probability that no point comes from $L'$ is at most $(1-\epsilon/4)^{O(1/\epsilon\log{(1/\epsilon)})}<\epsilon/100$. Applying union bound over all $O(1/\epsilon\log{(1/\epsilon)})$ points, with probability at least $9/10$, at least one point will come from the set $L'$. 
    \item Let $H'\subset H$ such that $H'=\Big\lbrace h' \in H | D_S(h')> \frac{1}{S}+\frac{\epsilon}{4S}\Big \rbrace$. Then $D_S(H')>\epsilon/4$. While sampling $O(\frac{1}{\epsilon}\log{(1/\epsilon)})$ points from $D_S$, using similar argument as that of the previous case, we argue that at least one point will occur from the set $H'$.
\end{itemize}
Let $i\in H'$ and $j\in L'$. Then $\frac{D_S(i)}{D_S(j)}>\frac{(1+\epsilon/4)}{(1-\epsilon/4)}>(1+\epsilon/2)$. Observe that $D_S(i)=\frac{w(i)}{W(S)}$ and $D_S(j)=\frac{w(j)}{W(S)}$. Therefore, when $D_S$ is far from uniformity, there exists at least a pair $(i,j)$ such that $\frac{w(i)}{w(j)}>(1+\epsilon/2)$ and the algorithm returns \reject in this case.
\end{proof}

As the elements of the universe are monotonically decreasing, we now connect it with monotone distributions. 

Let $D$ be monotone. By oblivious partitioning (Theorem \ref{thm: berge_1}) with parameter $\epsilon_1=\epsilon(1-\delta)$, for $\delta=\sqrt{(1-\epsilon)}$, we have $\sum_{j=1}^{\ell} \sum_{x\in I_j}|D(x)-\frac{D(I_j)}{|I_j|}|\leq \epsilon_1$. Or, $\sum_{j=1}^{\ell}D(I_j)\sum_{x\in I_j}|\frac{D(x)}{D(I_j)}-\frac{1}{|I_j|}|=\sum_{j=1}^{\ell} D(I_j) d_{TV}(D_{I_j},\mathcal{U}_{I_j})\leq \epsilon_1$. Consider $D$ be a distribution over $[n]$, where $D(x)=\frac{w(x)}{W}$, and $D_{I_j}$ is a conditional distribution over $I_j$ such that $D_{I_j}(x)=\frac{w(x)}{W(I_j)}$. Also, observe that $D(I_j)=\frac{W(I_j)}{W}$. Let $J$ be the set of intervals where for all $I_j$, $d_{TV}(D_{I_j},\mathcal{U}_{I_j})>\epsilon$, then $\sum_{I_j\in J} D(I_j) \leq (1-\delta)$. Consider $J'$ be the set of intervals where $D_{I_j}$ is not $\epsilon$-far from uniformity. Then $\sum_{I_j\in J'}D(I_j)>\delta$. In other words, we can consider $J'$ to be a large bucket where each point $i\in \{[n]\setminus J\}$. We also argue that when $D$ is monotone, $D(J')>\delta$. The core idea of our algorithm is to estimate the total weight of the elements that lie in the large bucket $J'$. The total weight of the points which lie in the intervals that are $\epsilon$-far from uniformity is at most $(1-\delta)$. Therefore, estimating the total weight of the set $J'$ will suffice to estimate the total weight of all elements upto $(1\pm \epsilon)$ multiplicative error. We present our algorithm with respect to a monotonically non-increasing universe over $[n]$ where for each $i\in[n]$, $D(i)=\frac{w(i)}{W}$ and $D_S(i)=\frac{w(i)}{W(S)}$.

\begin{algorithm}[hbt!]
    \caption{Estimate Sum Monotone}
    \label{alg:summono}
    \KwIn{Weighted conditioning sample access (over a set $S$) to a monotonically non-increasing universe, set of oblivious partitions $\{I_1, \dots, I_{\ell}\}$, for $\ell = O\left(\frac{1}{\epsilon_1}\log{n}\right)$ where $\epsilon_1 = \epsilon(1-\delta)$ for $\delta = \sqrt{(1-\epsilon)}$}
    \KwOut{Estimate $\hat{W}$ for the sum}
    
    $J = \emptyset$\;
    
    \For{each interval $I_j$}{
        Run Algorithm~\ref{alg:uniformity}\;
        \If{Output is \reject}{
            Add $I_j$ to $J$\;
        }
    }
    
    Sample $T = O\left(\frac{1}{\epsilon^6}\right)$ pairs $(i, w(i))$ (over $S = n$) from the universe\;
    
    \For{each $i \in T$}{
        $X_i = w(i)$; when $i \notin J$
    }
    
    Return $\hat{W} = \frac{1}{T} \sum_{i=1}^{T} X_i$\;
    \If{$J=\emptyset$}{
    Check \eIf{$w(1)=w(n)$}{Return $W=n\cdot w(1)$}
    {Sample a pair $(i,w(i))$ from each interval $I_j$ and return $W=\sum_{ j\in[\ell]}|I_j|\cdot w(i)$}
    
    }
    
\end{algorithm}
\begin{theorem}\label{thm: sum mono}
    The algorithm Estimate Sum Monotone uses total $O(\frac{1}{\epsilon^3}\log{n}+\frac{1}{\epsilon^6})$ conditional samples and returns an estimate $\hat{W}$ such that $(1-2\epsilon)W\leq \hat{W}\leq (1-\epsilon)W$.  
\end{theorem}
\begin{proof}
    Observe that the set $J$ consists of a set of intervals where the conditional distribution $D_{I_j}$ is $\epsilon$-far from uniform distribution over $I_j$. As the universe is monotone by the oblivious decomposition, we argue that $\sum_{I_j\in J} D(I_j) \leq (1-\delta)$. Considering $J'$ to be the set of intervals containing points from $[n]\setminus J$, $\sum_{I_j\in J'} D(I_j) >\delta$. Now, while sampling $O(\frac{1}{\epsilon^4})$ weighted samples from $[n]$, probability that no point comes from the set $J'$ is at most $(1-\delta)^{1/\epsilon^4}$. By applying a union bound over all samples, with high probability at least one point will come come the set $J'$. Now, we calculate the expectation of $\hat{W}$,
      $  \mathbb{E}[\hat{W}]=\sum_{i\in J'}\frac{W(J')}{W}\cdot w(i)
        =\frac{W^2(J')}{W}$.
\end{proof}
Now, we apply Hoeffding inequality (Lemma \ref{lem: hoeffding}) to bound the expectation. The random variables $X_i$ takes the value in the range $[\text{min} \ w(i), \text{max} \ w(i)]$. Therefore, we have
\begin{align*}
    Pr\big[ |\hat{W}-\frac{W^2(J')}{W}|>\epsilon\cdot \frac{W^2(J')}{W}\big]&\leq e^{-\frac{\epsilon^2W^4(J')T}{W^2\cdot[\text{max} \ w(i)- \text{min} \ w(i)]^2}}
\end{align*}
We know $\sum_{I_j\in J'}D(I_j)=D(J')>\delta$. In other words, $D(J')=\frac{W(J')}{W}>\delta$. Also, $[\text{max} \ w(i)- \text{min} \ w(i)]<W$. Together with the facts, we get that for $T=O(\frac{1}{\epsilon^6})$, with high probability $(1-\epsilon)\frac{W^2(J')}{W}\leq \hat{W}\leq (1+\epsilon)\frac{W^2(J')}{W}$. Now, $\delta^2W^2<W^2(J')<W^2$. For $\delta=\sqrt{(1-\epsilon)}$, we conclude that $(1-\epsilon)^2W<\hat{W}<(1+\epsilon)W$. Considering $(1-\epsilon)^2>(1-2\epsilon)$, finally we prove that $(1-2\epsilon)W<\hat{W}<(1+\epsilon)W$.
\begin{remark}
    Let the universe consists of the elements of all equal weights. In that case the distribution corresponding to the weights will be uniform distribution. If we use our algorithm for uniform distribution, while running the Algorithm \ref{alg:uniformity}, none of the intervals output \reject and the set $J$ will be empty. To get rid of such situation, we can make two queries, one for $w(1)$ and the other for $w(n)$. Given the promise that the elements are monotonically non-increasing, if $w(1)=w(n)$, we argue that $w(1)=w(2)=...=w(n)$. Therefore, in this case, we return the weight $W=n\cdot w(1)$. Similarly, another case would be that the set $J$ remains empty at the end of the algorithm but $w(1)\neq w(n)$. This indicates that the conditional distribution over each interval is uniform over $I_j$. In this case, we sample any point $(i,w(i))$ from each interval and return $W=\sum_{ j\in[\ell]}|I_j|\cdot w(i)$. 
\end{remark}

\subsection{A Special Case of Unimodal Distribution}
Let the weights of the elements of the universe be monotonically decreasing till a point $j$ and then it monotonically increasing till $n$. In other words the distribution corresponding to the weights is unimodal.  Unimodal distribution can be thought of as union of two monotone distributions where one is monotonically decreasing and the other is monotonically increasing. Let the first half of the distribution is monotonically decreasing. If we could know the lowest point of the entire distribution, we will be sure that from thereon the next half of the distribution is monotonically increasing. Then apply our algorithm for two monotone distributions separately. Let $\hat{W}_1$, and $\hat{W}_2$ be the estimated sum of the two halves of the distribution respectively. Then the estimated sum is $\hat{W}_1+\hat{W}_2$.

To find the point with the lowest weight, we use the concept of binary search. We start with querying the weight of the middle element $w(n/2)$ and compare it with $w(n/2-1)$ and $w(n/2+1)$. If $w(n/2-1)<w(n/2)<w(n/2+1)$, then we are in the increasing part. Therefore the lowest point must be on the interval $[1,n/2-1]$, we remove the interval $[n/2,n]$ and recurse on the left hand side. Similarly, if $w(n/2-1)>w(n/2)>w(n/2+1)$ then we are still on the decreasing part and the minimum point must be sitting on the right hand side and we recurse on the right side. Repeat this process until there are constant number of elements in an interval. Then query all this points and find the minimum point. This procedure requires $O(\log{n})$ additional queries. We argue that we do not sample anything from $U$ in this process, rather we ask the the weight $w(i)$ for an input $i$. This queries are identical to making \eval queries which is less powerful than \emph{dual model}. Therefore, we argue that our procedure needs $O(\log{n})$ additional \eval queries which are adaptive in nature.


Once we find the lowest point, we use the Algorithm \ref{alg:summono} for the two halves of the universe separately. As the first half of the distribution is decreasing, we use the oblivious decomposition according to the Theorem \ref{thm: berge_1}, and for the second half of the distribution, we use the Theorem \ref{thm: berge_2}. Let $W_1$ and $W_2$ be the weight of the first and second half of the universe respectively and $\hat{W}_1$ and $\hat{W}_2$ are the corresponding estimated weights. Applying the Theorem \ref{thm: sum mono}, we argue that $(1-2\epsilon)W_1\leq \hat{W}_1\leq (1-\epsilon)W_1$ and $(1-2\epsilon)W_2\leq \hat{W}_2\leq (1-\epsilon)W_2$. Hence, $(1-2\epsilon)W\leq \hat{W}_1+\hat{W}_2\leq (1-\epsilon)W$. We formalize the theorem below,

\begin{theorem}\label{thm: sum unimodal}
   Let there is a universe of $[n]$ numbers where the weights of the points are decreasing up to a point and then increasing till $[n]$. There exists an algorithm that uses total $O(\frac{1}{\epsilon^3}\log{n}+\frac{1}{\epsilon^6})$ non-adaptive conditional samples and returns an estimate $\hat{W}$ such that $(1-2\epsilon)W\leq \hat{W}\leq (1-\epsilon)W$. The algorithm additionally requires $O(\log{n})$ adaptive evaluation queries to find the index of the minimum weighted element. 
\end{theorem}

\subsection{Discussion on Lower Bound}
    We observe that Acharya et al. \cite{ACK14} proves an $\Omega(\frac{\log{n}}{\log{\log{n}}})$ lower bound for estimating support size up to a factor $\log{n}$ in the conditional sampling model. In particular, they considered non-adaptive conditional queries in this scenario. Later, in a subsequent work Acharya et al. \cite{ACK15} use the same techniques and adapt their argument for a smaller (constant) value of the approximation factor. They claim the following,
    \begin{theorem}\cite{ACK15}
        Any non-adaptive algorithm for estimating $W_S$
up to a factor $2$ requires at least $\Omega(\log{n})$ queries.
    \end{theorem}
Our algorithm for estimating sum for monotonically non-increasing universe uses $O(\frac{1}{\epsilon^3}\log{n}+\frac{1}{\epsilon^6})$ non-adaptive conditional samples. Therefore, We argue that our Algorithm \ref{alg:summono} follows the non-adaptive lower bound proposed by \cite{ACK15} in terms of $\log{n}$.

\section{Estimating support size of the universe using hybrid model}
Let there are $n$ elements in the universe $U$ where each element $i$ is associated with weight $w(i)$ such that every element of the domain either no weight or at least some minimum weight $w(i)\geq \frac{W}{n}$ where $W$ is the total weight of the elements. Let the support size of the universe be $k\subset n$, which is defined by the number of non-zero elements on the domain. The goal is to estimate $k$ up to an error of $\pm \epsilon n$. The model we are considering in this context is the hybrid model where we are allowed to sample a pair $(i,w(i))$ uniformly at random as well as using weighted sampling. 

We define a procedure called \emph{estimate neighborhood fraction}, which takes a point $(x,w(x))$ as input and returns the approximate fraction of point lie in the neighborhood of $x$. For an element $x\in[n]$, a neighborhood of $x$ is defined as a set $U_\epsilon(x)=\{y\in[n]: \frac{1}{1+\epsilon}w(x)\leq w(y)\leq (1+\epsilon)w(x)\}$. Intuitively, the procedure returns an estimate of $\frac{|U_{\epsilon}(x)|}{n}$.

\begin{algorithm}[hbt!]
    \caption{Estimate Neighborhood Fraction}
    \label{alg:estimate_neigh}
    \KwIn{A pair $(i, w(i))$, $\epsilon, \alpha \in (0, 1)$, where $\alpha = \frac{\epsilon^3}{\log(n/\epsilon)\log(\log{n/\epsilon})}$}
    \KwOut{Estimated neighborhood fraction $\hat{f}$}
    
    $X = 0$\;
    
    Sample $S = O\left(\frac{1}{\epsilon_1^2}\right)$ set of $(j, w(j))$ pairs uniformly from $[n]$\;
    
    \For{each $(j, w(j))$ pair}{
        \If{$\frac{w(i)}{(1 + \epsilon)} \leq w(j) \leq (1 + \epsilon) w(i)$}{
            Set $X_j = 1$\;
        }
        $X = X + X_j$\;
    }
    
    Return $\hat{f} = \frac{1}{S} \cdot X$\;
    
\end{algorithm}
\begin{theorem}
  Let $U_\epsilon(i)=\{j\in[n]: \frac{1}{1+\epsilon}w(i)\leq w(j)\leq (1+\epsilon)w(i)\}$ be the neighborhood of the point $i$. The algorithm Estimate Neighborhood Fraction returns $\hat{f}$, such that $\frac{|U_{\epsilon}(i)|}{n}-\frac{\epsilon^3}{\log{(n/\epsilon)}\log{\log{n/\epsilon}}}\leq \hat{f}\leq \frac{|U_{\epsilon}(i)|}{n}+\frac{\epsilon^3}{\log{(n/\epsilon)}\log{\log{n/\epsilon}}}$. The algorithm uses $O(\frac{\log^2{(n/\epsilon)}\log^2{\log{n/\epsilon}}}{\epsilon^6})$ uniform samples from $[n]$. 
\end{theorem}
\begin{proof}
    $\hat{f}=\frac{\sum_{j\in S} X_j}{S}$. Hence, $\mathbb{E}[\hat{f}]=\mathbb{E}[X_j]=\sum_{j; j\in U_{\epsilon}(i)} \frac{1}{n}=\frac{|U_{\epsilon}(i)|}{n}$.

    $X_1,X_2,...,X_S$ are all random variables taking values in $[0,1]$. For $S=O(\frac{1}{\alpha^2})=O(\frac{\log^2{(n/\epsilon)}\log^2{\log{n/\epsilon}}}{\epsilon^6})$, by applying Chernoff bound, we get
    \begin{align*}
        Pr[|\hat{f}-\frac{|U_{\epsilon}(i)|}{n}|>\alpha]&\leq e^{-\alpha^2.S}
        <1/10.
    \end{align*}
    Hence, with probability at least $9/10$, $|\hat{f}-\frac{|U_{\epsilon}(i)|}{n}|\leq \frac{\epsilon^3}{\log{(n/\epsilon)}\log{\log{n/\epsilon}}}$.
\end{proof}

The following definition is useful to understand the cover of a set with respect to a distribution $D_S$,
\begin{definition}\label{def: weight_cover}[Weight Cover \cite{CRS15}]
    Let $D$ be a distribution over $[n]$ and a parameter $\epsilon_1>0$, we say that a point $i\in[n]$ is $\epsilon_1$-covered by a set $R=\{r_1,...,r_t\}\subseteq [n]$ if there exists a point $r_j\in R$ such that $D(i)\in [1/(1+\epsilon_1), (1+\epsilon_1)]D(r_j)$. Let the set of points in $[n]$ that are $\epsilon_1$-covered by $R$ be denoted by $U_{\epsilon_1}(R)$. We say that $R$ is an $(\epsilon_1,\epsilon_2)$-cover for $D$ if $D([n]\setminus U_{\epsilon_1})\leq \epsilon_2$.
\end{definition}
Considering $D(i)=\frac{w(i)}{W}$ in our case, analogous to the above definition, we say that an element $i\in [n]$ is $\epsilon_1$-cover by a set $R=\{r_1,...,r_t\}\subseteq [n]$, if there exists an element $r_j\in R$ such that $w(i)\in [1/(1+\epsilon_1), (1+\epsilon_1)]w(r_j)$. Similarly, if $U_{\epsilon_1}(R)$ be the set of points $\epsilon_1$-covered by $R$, then $W([n]\setminus U_{\epsilon_1})\leq \epsilon_2\cdot W$. We define the following lemma related to the cover of a set,
\begin{lemma}\label{lem: weight_cover}\cite{CRS15}
    Let $D$ be a distribution over $[n]$. Given any fixed $c>0$, there exists a constant $c'>0$, such that with probability at least $99/100$, a sample $R$ of size $m=O\big( \frac{\log{(n/\epsilon)}}{\epsilon^2}\cdot \log{\frac{\log{(n/\epsilon)}}{\epsilon}}\big)$ drawn according to distribution $D_S$ is an $(\epsilon/c,\epsilon/c)$-cover for $D_S$.
\end{lemma}
The above lemma is applicable in our case as well when $D(i)=\frac{w(i)}{W}$. Thus, we say that while sampling a set of points $R=O\big( \frac{\log{(n/\epsilon)}}{\epsilon^2}\cdot \log{\frac{\log{(n/\epsilon)}}{\epsilon}}\big)$ points according to the weighted sampling from the universe $U$, $R$ is an $(\epsilon/c,\epsilon/c)$-cover for $D$. We formalize the statement in the lemma below,
\begin{lemma}\label{lem: weightcover_new}
    Let there is a set of size $|U|=n$, where each element $i$ has associated with weight $w(i)$. Let $R=O\big( \frac{\log{(n/\epsilon)}}{\epsilon^2}\cdot \log{\frac{\log{(n/\epsilon)}}{\epsilon}}\big)$ samples are drawn according to the proportional sampling. If $U_{\epsilon/c}(R)$ be the set of points $\epsilon/c$-covered by $R$, then $W([n]\setminus U_{\epsilon_1}(R))\leq (\epsilon/c)\cdot W$.
\end{lemma}
\begin{algorithm}[hbt!]
    \caption{Estimate Support Size}
    \label{alg:estimate_sup}
    \KwIn{A set of weighted pairs $(i, w(i))$, $\epsilon \in (0, 1)$}
    \KwOut{Estimated support size $\hat{k}$}

    Sample $R = O\left( \frac{\log{(n/\epsilon)}}{\epsilon^2} \cdot \log{\frac{\log{(n/\epsilon)}}{\epsilon}} \right)$ set of $(i, w(i))$ pairs according to weighted sampling\;
    
    Set $\hat{k} = 0$\;
    
    \For{each $i \in R$}{
        Run Algorithm~\ref{alg:estimate_neigh} and let $\hat{f}_i$ be the output\;
        Calculate $\hat{k} = \hat{k} + n \cdot \hat{f}_i$\;
    }
    
    Return $\hat{k}$\;
    
\end{algorithm}
\begin{theorem}
 Given a threshold $n$ and access to an universe such that $\text{min}_{x\in \text{sup}}w(i)\geq \frac{W}{n}$, the algorithm Estimate Support Size returns an estimate of the support size such that $k-2\epsilon n\leq\hat{k}\leq k+\epsilon n$. The algorithm uses $O\big( \frac{\log^3{(n/\epsilon)}}{\epsilon^8}\cdot \log^4{\frac{\log{(n/\epsilon)}}{\epsilon}}\big)$ uniform samples and $O\big( \frac{\log{(n/\epsilon)}}{\epsilon^2}\cdot \log{\frac{\log{(n/\epsilon)}}{\epsilon}}\big)$ weighted samples from the universe.   
\end{theorem}

\begin{proof}
    Observe that $\hat{k}=\sum_{i\in R}n\cdot \hat{f}_i$. Also, $\frac{|U_{\epsilon}(i)|}{n}-\frac{\epsilon^3}{\log{(n/\epsilon)}\log{\log{n/\epsilon}}}\leq \hat{f}_i\leq \frac{|U_{\epsilon}(i)|}{n}+\frac{\epsilon^3}{\log{(n/\epsilon)}\log{\log{n/\epsilon}}}$. Summing over $R$, we get,
    \begin{align*}
   \sum_{i\in R}|U_{\epsilon}(i)|-\epsilon\cdot n\leq    \hat{k}\leq \sum_{i\in R}|U_{\epsilon}(i)|+\epsilon\cdot n
    \end{align*}
We know that by Lemma \ref{lem: weightcover_new}, $W([n]\setminus U_{\epsilon}(R))\leq (\epsilon/c)\cdot W$. As $[k]$ defines the set of elements with non-zero weights, we argue that $W([k]\setminus U_{\epsilon}(R))\leq (\epsilon/c)\cdot W$. Let $T$ be the set of elements that are not $\epsilon/c$ covered by $R$ such that for all $i\in T$, $w(i)>0$. Then by the Lemma \ref{lem: weightcover_new}, $W(T)\leq (\epsilon/c)\cdot W$. We are also given the promise that for $i\in T$, $\text{min}_{x\in \text{sup}}w(i)\geq \frac{W}{n}$. Therefore, $|T|\cdot \frac{W}{n}\leq \frac{\epsilon W}{c}$, or, $|T|\leq n\epsilon/c\leq n\epsilon$ for $c\geq 1$. Now, observe that $\sum_{i\in R}|U_{\epsilon}(i)|=k-|T|\geq k-n\epsilon$. Additionally, $\sum_{i\in R}|U_{\epsilon}(i)|\leq k$. Substituting these values, finally we get
  $k-2\epsilon\cdot n\leq  \hat{k}\leq k+\epsilon\cdot n$.

    \textbf{Sample complexity analysis:} The algorithm samples $O\big( \frac{\log{(n/\epsilon)}}{\epsilon^2}\cdot \log{\frac{\log{(n/\epsilon)}}{\epsilon}}\big)$ weighted points from the universe. For each point the algorithm runs "Estimate neighborhood fraction". Estimating neighborhood for a fixed point $i$ requires $O(\frac{\log^2{(n/\epsilon)}\log^2{\log{n/\epsilon}}}{\epsilon^6})$ uniform samples from $[n]$. To guarantee that the estimated neighborhood is a good estimator for every $i\in R$ points, we need to use a union bound over $R$. Hence, each run of "Estimate neighborhood fraction" uses $O(\frac{\log^2{(n/\epsilon)}\log^3{\log{n/\epsilon}}}{\epsilon^6})$. Therefore, a total of  $O\big( \frac{\log^3{(n/\epsilon)}}{\epsilon^8}\cdot \log^4{\frac{\log{(n/\epsilon)}}{\epsilon}}\big)$ uniform samples required by the algorithm.
\end{proof}
\section{Conclusion}
We have presented an algorithm for estimating the sum of \( n \) weighted elements when the underlying distribution \( D \) corresponding to the weights is known to be monotone, leveraging both weighted and uniform conditional queries. Additionally, we extended the problem to the case where \( D \) is unimodal, providing an algorithm tailored for this scenario. An intriguing direction for future work is to explore the sum estimation problem when \( D \) is \( k \)-modal, a histogram, or follows other complex distributions. Furthermore, investigating sum estimation using only weighted conditional samples remains an open challenge. We also addressed the problem of approximating the support size when the weight of each element is either zero or at least \( W/n \), employing a hybrid query model that combines weighted and uniform sampling from \( U \). Another promising avenue for future research is to consider solving this problem using only weighted sampling. 

\section{Acknowledgment}
We would like to thank Dr. Anup Bhattacharya and Dr. Yadu Vasudev, for their helpful guidance and discussions.

\newpage
\bibliographystyle{plainurl}
\bibliography{ref}

\begin{thebibliography}{10}

\bibitem{ACK14}
Jayadev Acharya, Cl\'{e}ment~L. Canonne, and Gautam Kamath.
\newblock {A Chasm Between Identity and Equivalence Testing with Conditional Queries}.
\newblock In {\em Approximation, Randomization, and Combinatorial Optimization. Algorithms and Techniques (APPROX/RANDOM 2015)}, volume~40 of {\em Leibniz International Proceedings in Informatics (LIPIcs)}, pages 449--466, 2015.
\newblock \href {https://doi.org/10.4230/LIPIcs.APPROX-RANDOM.2015.449} {\path{doi:10.4230/LIPIcs.APPROX-RANDOM.2015.449}}.

\bibitem{ACK15}
Jayadev Acharya, Clément~L. Canonne, and Gautam Kamath.
\newblock Adaptive estimation in weighted group testing.
\newblock In {\em 2015 IEEE International Symposium on Information Theory (ISIT)}, pages 2116--2120, 2015.
\newblock \href {https://doi.org/10.1109/ISIT.2015.7282829} {\path{doi:10.1109/ISIT.2015.7282829}}.

\bibitem{BKS01}
Ziv Bar-Yossef, Ravi Kumar, and D.~Sivakumar.
\newblock Sampling algorithms: lower bounds and applications.
\newblock In {\em Proceedings of the Thirty-Third Annual ACM Symposium on Theory of Computing}, STOC '01, page 266–275, New York, NY, USA, 2001. Association for Computing Machinery.
\newblock \href {https://doi.org/10.1145/380752.380810} {\path{doi:10.1145/380752.380810}}.

\bibitem{BFF+01}
T.~Batu, L.~Fortnow, E.~Fischer, R.~Kumar, R.~Rubinfeld, and P.~White.
\newblock Testing random variables for independence and identity.
\newblock In {\em Proceedings of the 42Nd IEEE Symposium on Foundations of Computer Science}, FOCS '01, pages 442--, Washington, DC, USA, 2001. IEEE Computer Society.

\bibitem{BKR04}
Tugkan Batu, Ravi Kumar, and Ronitt Rubinfeld.
\newblock Sublinear algorithms for testing monotone and unimodal distributions.
\newblock In {\em Proceedings of the Thirty-sixth Annual ACM Symposium on Theory of Computing}, STOC '04, pages 381--390, New York, NY, USA, 2004. ACM.

\bibitem{BFR+10}
Tu\u{g}kan Batu, Lance Fortnow, Ronitt Rubinfeld, Warren~D. Smith, and Patrick White.
\newblock Testing closeness of discrete distributions.
\newblock {\em J. ACM}, 60(1), feb 2013.

\bibitem{BT24}
Lorenzo Beretta and Jakub T\v{e}tek.
\newblock Better sum estimation via weighted sampling.
\newblock {\em ACM Trans. Algorithms}, 20(3), June 2024.
\newblock \href {https://doi.org/10.1145/3650030} {\path{doi:10.1145/3650030}}.

\bibitem{BY22}
Arnab Bhattacharyya and Yuichi Yoshida.
\newblock {\em Property Testing - Problems and Techniques}.
\newblock Springer, 2022.
\newblock \href {https://doi.org/10.1007/978-981-16-8622-1} {\path{doi:10.1007/978-981-16-8622-1}}.

\bibitem{Birge}
Lucien Birge.
\newblock {On the Risk of Histograms for Estimating Decreasing Densities}.
\newblock {\em The Annals of Statistics}, 15(3):1013 -- 1022, 1987.

\bibitem{CEG95}
Ran Canetti, Guy Even, and Oded Goldreich.
\newblock Lower bounds for sampling algorithms for estimating the average.
\newblock {\em Inf. Process. Lett.}, 53(1):17–25, January 1995.
\newblock \href {https://doi.org/10.1016/0020-0190(94)00171-T} {\path{doi:10.1016/0020-0190(94)00171-T}}.

\bibitem{CR14}
Cl{\'e}ment Canonne and Ronitt Rubinfeld.
\newblock Testing probability distributions underlying aggregated data.
\newblock In {\em Automata, Languages, and Programming}, pages 283--295, Berlin, Heidelberg, 2014. Springer Berlin Heidelberg.

\bibitem{Can15}
Cl\'ement~L. Canonne.
\newblock {B}ig {D}ata on the rise: {T}esting monotonicity of distributions.
\newblock In {\em 42nd International Conference on Automata, Languages and Programming (ICALP)}, 2015.

\bibitem{CDGR16}
{Cl{\'e}ment L.} Canonne, Ilias Diakonikolas, Themis Gouleakis, and Ronitt Rubinfeld.
\newblock Testing shape restrictions of discrete distributions.
\newblock In {\em Proceedings of the 33rd International Symposium on Theoretical Aspects of Computer Science (STACS 2016)}, Leibniz International Proceedings in Informatics (LIPIcs), pages 1--14. Schloss Dagstuhl - Leibniz-Zentrum fuer Informatik, Germany, 2016.

\bibitem{CRS15}
Cl\'{e}ment~L. Canonne, Dana Ron, and Rocco~A. Servedio.
\newblock Testing probability distributions using conditional samples.
\newblock {\em SIAM Journal on Computing}, 44(3):540--616, 2015.

\bibitem{CKM23}
Diptarka Chakraborty, Gunjan Kumar, and Kuldeep~S. Meel.
\newblock {Support Size Estimation: The Power of Conditioning}.
\newblock In {\em 48th International Symposium on Mathematical Foundations of Computer Science (MFCS 2023)}, volume 272 of {\em Leibniz International Proceedings in Informatics (LIPIcs)}, pages 33:1--33:13. Schloss Dagstuhl -- Leibniz-Zentrum f{\"u}r Informatik, 2023.
\newblock \href {https://doi.org/10.4230/LIPIcs.MFCS.2023.33} {\path{doi:10.4230/LIPIcs.MFCS.2023.33}}.

\bibitem{CFGM13}
Sourav Chakraborty, Eldar Fischer, Yonatan Goldhirsh, and Arie Matsliah.
\newblock On the power of conditional samples in distribution testing.
\newblock {\em SIAM Journal on Computing}, 45(4):1261--1296, 2016.

\bibitem{CDSS13}
Siu-On Chan, Ilias Diakonikolas, Rocco~A. Servedio, and Xiaorui Sun.
\newblock Learning mixtures of structured distributions over discrete domains.
\newblock In {\em Proceedings of the Twenty-Fourth Annual ACM-SIAM Symposium on Discrete Algorithms}, SODA '13, page 1380–1394, USA, 2013. Society for Industrial and Applied Mathematics.

\bibitem{CDVV14}
Siu-On Chan, Ilias Diakonikolas, Paul Valiant, and Gregory Valiant.
\newblock Optimal algorithms for testing closeness of discrete distributions.
\newblock In {\em Proceedings of the 2014 Annual ACM-SIAM Symposium on Discrete Algorithms (SODA)}, pages 1193--1203, 2014.

\bibitem{CCMN00}
Moses Charikar, Surajit Chaudhuri, Rajeev Motwani, and Vivek Narasayya.
\newblock Towards estimation error guarantees for distinct values.
\newblock In {\em Proceedings of the Nineteenth ACM SIGMOD-SIGACT-SIGART Symposium on Principles of Database Systems}, PODS '00, page 268–279, New York, NY, USA, 2000. Association for Computing Machinery.
\newblock \href {https://doi.org/10.1145/335168.335230} {\path{doi:10.1145/335168.335230}}.

\bibitem{DDS12}
Constantinos Daskalakis, Ilias Diakonikolas, and Rocco~A. Servedio.
\newblock Learning k-modal distributions via testing.
\newblock {\em Theory Comput.}, 10:535--570, 2012.
\newblock URL: \url{https://api.semanticscholar.org/CorpusID:7739848}.

\bibitem{DDSVV13}
Constantinos Daskalakis, Ilias Diakonikolas, {Rocco A.} Servedio, Gregory Valiant, and Paul Valiant.
\newblock Testing k-modal distributions: Optimal algorithms via reductions.
\newblock Conference, 24th Annual ACM-SIAM Symposium on Discrete Algorithms, SODA, 2013.

\bibitem{DKN15a}
Ilias Diakonikolas, Daniel~M. Kane, and Vladimir Nikishkin.
\newblock Optimal algorithms and lower bounds for testing closeness of structured distributions.
\newblock In {\em 2015 IEEE 56th Annual Symposium on Foundations of Computer Science}, pages 1183--1202, 2015.
\newblock \href {https://doi.org/10.1109/FOCS.2015.76} {\path{doi:10.1109/FOCS.2015.76}}.

\bibitem{DKN15b}
Ilias Diakonikolas, Daniel~M. Kane, and Vladimir Nikishkin.
\newblock Testing identity of structured distributions.
\newblock In {\em Proceedings of the Twenty-Sixth Annual ACM-SIAM Symposium on Discrete Algorithms}, SODA '15, page 1841–1854, USA, 2015. Society for Industrial and Applied Mathematics.

\bibitem{DKN15}
Ilias Diakonikolas, Daniel~M. Kane, and Vladimir Nikishkin.
\newblock Testing identity of structured distributions.
\newblock In {\em Proceedings of the Twenty-Sixth Annual ACM-SIAM Symposium on Discrete Algorithms}, SODA '15, page 1841–1854, USA, 2015. Society for Industrial and Applied Mathematics.

\bibitem{Fischer01}
Eldar Fischer.
\newblock The art of uninformed decisions.
\newblock {\em Bull. {EATCS}}, 75:97, 2001.

\bibitem{FLV19}
Eldar Fischer, Oded Lachish, and Yadu Vasudev.
\newblock Improving and extending the testing of distributions for shape-restricted properties.
\newblock {\em Algorithmica, Springer}, 81,3765–3802, 2019.

\bibitem{Goldreich10}
Oded Goldreich, editor.
\newblock {\em Property Testing - Current Research and Surveys}, volume 6390 of {\em Lecture Notes in Computer Science}.
\newblock Springer, 2010.

\bibitem{Goldreich17}
Oded Goldreich.
\newblock {\em Introduction to Property Testing}.
\newblock 2017.

\bibitem{GR00}
Oded Goldreich and Dana Ron.
\newblock On testing expansion in bounded-degree graphs.
\newblock {\em Electron Colloq Comput Complexity}, 7, 01 2000.

\bibitem{MPX07}
Rajeev Motwani, Rina Panigrahy, and Ying Xu.
\newblock Estimating sum by weighted sampling.
\newblock In {\em Proceedings of the 34th International Conference on Automata, Languages and Programming}, ICALP'07, page 53–64, Berlin, Heidelberg, 2007. Springer-Verlag.

\bibitem{Pan08}
L.~Paninski.
\newblock A coincidence-based test for uniformity given very sparsely sampled discrete data.
\newblock {\em IEEE Trans. Inf. Theor.}, 54(10):4750--4755, October 2008.

\bibitem{RRSS09}
Sofya Raskhodnikova, Dana Ron, Amir Shpilka, and Adam Smith.
\newblock Strong lower bounds for approximating distribution support size and the distinct elements problem.
\newblock {\em SIAM Journal on Computing}, 39(3):813--842, 2009.
\newblock \href {https://doi.org/10.1137/070701649} {\path{doi:10.1137/070701649}}.

\bibitem{RS05}
Ronitt Rubinfeld and Rocco~A. Servedio.
\newblock Testing monotone high-dimensional distributions.
\newblock In {\em Proceedings of the Thirty-Seventh Annual ACM Symposium on Theory of Computing}, STOC '05, page 147–156, New York, NY, USA, 2005. Association for Computing Machinery.

\bibitem{RV20}
Ronitt Rubinfeld and Arsen Vasilyan.
\newblock {Monotone Probability Distributions over the Boolean Cube Can Be Learned with Sublinear Samples}.
\newblock In {\em 11th Innovations in Theoretical Computer Science Conference (ITCS 2020)}, volume 151, pages 28:1--28:34. Schloss Dagstuhl -- Leibniz-Zentrum f{\"u}r Informatik, 2020.

\bibitem{VV10}
Gregory Valiant and Paul Valiant.
\newblock A clt and tight lower bounds for estimating entropy.
\newblock {\em Electronic Colloquium on Computational Complexity (ECCC)}, 17:183, 01 2010.

\bibitem{VV14}
Gregory Valiant and Paul Valiant.
\newblock An automatic inequality prover and instance optimal identity testing.
\newblock In {\em Proceedings of the 2014 IEEE 55th Annual Symposium on Foundations of Computer Science}, FOCS '14, pages 51--60, Washington, DC, USA, 2014. IEEE Computer Society.

\bibitem{Valiant}
Paul Valiant.
\newblock Testing symmetric properties of distributions.
\newblock {\em SIAM Journal on Computing}, 40(6):1927--1968, 2011.

\end{thebibliography}
\newpage

\end{document}